\documentclass[a4paper,11pt]{llncs}

\usepackage{float}

 \usepackage{graphicx}
\usepackage[T1]{fontenc}
\usepackage{url}
\usepackage{mathtools}
\usepackage{nicefrac}
\usepackage{tcolorbox}
\usepackage{braket}
\usepackage{booktabs}

\usepackage{xcolor}
\definecolor{darkgreen}{rgb}{0.0, 0.5, 0.0}  

 \usepackage{framed, color}
\definecolor{shadecolor}{rgb}{0.79,0.87,1}

\usepackage{amsmath}
\usepackage{amssymb}                        
\usepackage{amsfonts}
\usepackage{mathtools}
\usepackage{pstricks}
\usepackage{pst-all}
\usepackage{multirow}
\usepackage{slashbox}
\usepackage{url}
\usepackage{enumitem}
\usepackage{tikz}
\usetikzlibrary{shapes,arrows}
\usetikzlibrary{shapes, arrows, positioning, calc}

\def\thebibliography#1{\section*{References\markboth
 {REFERENCES}{REFERENCES}}\list
 {[\arabic{enumi}]}{\settowidth\labelwidth{[#1]}\leftmargin\labelwidth
 \advance\leftmargin\labelsep
 \usecounter{enumi}}
 \def\newblock{\hskip .11em plus .33em minus -.07em}
 \sloppy
 \sfcode`\.=1000\relax}

\usepackage{ntheorem}
\theorembodyfont{\normalfont}

\newtheorem{assumption}{Assumption}

\usepackage[hidelinks]{hyperref}
\usepackage{bookmark}

\usepackage[utf8]{inputenc}
\usepackage{algorithm}
\usepackage{algpseudocode}

\title{Synchronization-Free Algebraic Fingerprints for Large Language Models: 
From Autoregressive to Diffusion Models}

\begin{document}
\pagestyle{headings}
\renewcommand{\labelitemi}{$\bullet$}
\author{
Jaros\l{}aw Janas\inst{1} \and
Josef Pieprzyk\inst{1,2} \and
Pawe\l{} Morawiecki\inst{1} 
}
\institute{
Institute of Computer Science, Polish Academy of Sciences, Warsaw, Poland \and
CSIRO, Sydney, Australia
}
\maketitle

\begin{abstract}
Large Language Models (LLMs) have created an urgent need for reliable watermarking methods that enable attribution of generated text while remaining robust to editing and paraphrasing. We propose a novel synchronization-free watermarking scheme in which every watermark consists of a single binary congruence generated from a pair of neighbouring tokens. For each token pair, a cryptographic hash determines an evaluation point of a Reed--Solomon polynomial representing the secret identity, while the parity of the polynomial evaluation determines the watermark bit embedded into the second token of the pair. Since each congruence is self-contained and depends only on the local token pair, the proposed construction is naturally resistant to insertions, deletions, and token reordering. We analyse the recovery problem from an algebraic perspective, discuss several decoding algorithms suitable for different identity sizes, and model watermark corruption as a Binary Symmetric Channel. The analysis shows that reliable recovery requires only a small redundancy even for relatively high token corruption rates. Unlike existing block-based watermarking schemes, the proposed method avoids synchronization problems while providing a flexible framework for embedding both short and long secret identities.
\end{abstract}


\section{Introduction}

Artificial Intelligence (AI) has become one of the most influential technologies of the twenty-first century, transforming science, industry, healthcare, finance, education, and cybersecurity. Recent advances in deep learning have enabled machines to perform tasks that were previously considered uniquely human, including image recognition, speech understanding, code generation, and natural language processing. Among these developments, Large Language Models (LLMs) represent a major breakthrough. Trained on massive text corpora and based primarily on the Transformer architecture, LLMs are capable of generating coherent, contextually relevant, and human-like text for a wide variety of applications, including question answering, summarization, software development, scientific writing, and conversational AI. Representative examples include GPT-4~\cite{openai2023gpt4}, PaLM~\cite{chowdhery_2022}, and LLaMA~\cite{touvron_2023}.

The widespread deployment of LLMs has also created new security, privacy, and intellectual property challenges. As AI-generated content becomes increasingly indistinguishable from human writing, it is becoming difficult to determine its origin or to establish ownership. This raises important concerns regarding copyright protection, attribution of authorship, accountability for generated content, academic integrity, automated misinformation, and malicious use of AI. Consequently, techniques that allow AI-generated content to be authenticated or traced back to its source are becoming an essential component of trustworthy AI systems.

Digital watermarking offers one of the most promising approaches to addressing these challenges. Traditionally, watermarking has been extensively studied for digital images, audio, video, software, databases, and multimedia content, where it has been used for copyright protection, ownership verification, traitor tracing, fingerprinting, tamper detection, and content authentication. More recently, the concept has been extended to generative AI, where the objective is to embed an invisible identifier into generated text while preserving its semantic quality and fluency. Such identifiers enable the origin of the generated content to be verified, facilitate forensic investigations, and provide a mechanism for assigning responsibility to individual users or AI systems.

Text watermarking for LLMs has recently become an active research area. Existing approaches may be broadly classified into statistical watermarking, lexical or semantic substitution methods, syntactic watermarking, and cryptographic watermarking schemes. Among these, statistical approaches based on biased token sampling have received considerable attention due to their simplicity and compatibility with existing autoregressive language models~\cite{kirchenbauer2023watermark}. More sophisticated constructions employ error-correcting codes, secret-key token partitioning, or constrained decoding to improve robustness against editing attacks. Nevertheless, most existing schemes remain inherently synchronization dependent. Their encoding process partitions the generated text into fixed-size blocks or relies on the sequential position of tokens, making watermark recovery sensitive to insertions, deletions, sentence reordering, and other editing operations that destroy synchronization between the encoder and decoder. Developing synchronization-free watermarking techniques capable of reliable detection after arbitrary text editing therefore remains one of the fundamental open problems in LLM watermarking.

In this paper we propose a synchronization-free watermarking framework that eliminates the need for explicit token alignment during watermark recovery. Instead of relying on the positional structure of the generated text, the proposed construction embeds the watermark into statistically independent token pairs selected according to a secret cryptographic key. The resulting scheme naturally supports probabilistic analysis through a Binary Symmetric Channel model, allowing the robustness of watermark recovery to be evaluated analytically rather than relying exclusively on empirical experimentation.

\section{Related Work and Motivation}

Digital watermarking has been extensively studied for more than three decades as a mechanism for copyright protection, ownership verification, fingerprinting, tamper detection, and content authentication~\cite{cox2007}. Classical watermarking techniques were originally developed for multimedia objects such as images, audio, and video, where a hidden signal is embedded while preserving perceptual quality. With the rapid adoption of Large Language Models (LLMs), watermarking has become equally important for protecting AI-generated content, enabling provenance verification, authorship attribution, misuse detection, and accountability.

Existing watermarking techniques for LLMs may be broadly classified into four categories. The first class embeds the watermark during model training by modifying either the training data or the optimization objective. Representative examples include Text Radioactivity~\cite{Sander_2024}, TextMarker~\cite{Liu_2023}, and knowledge injection techniques~\cite{Cui_2025}. While these approaches provide strong ownership guarantees, they require retraining or fine-tuning the underlying model, making them impractical for closed-source commercial LLMs.

The second category modifies the token probability distribution during text generation. The pioneering work of Kirchenbauer \emph{et al.}~\cite{kirchenbauer2023watermark} partitions the vocabulary into green and red token sets and softly biases the sampling process toward green tokens, creating statistically detectable watermarks. Numerous variants have since been proposed, including logits-to-text watermarking~\cite{wong_2025}, GumbelSoft watermarking~\cite{Fu_2024}, semantic watermarking~\cite{Kuditipudi2024}, and SynthID-Text developed by Google DeepMind~\cite{Nature2024SynthID}. These methods preserve text quality while achieving reliable statistical detection.

A third class embeds watermarks directly during token sampling by introducing controlled randomness or lexical constraints~\cite{Xu_2024,Zhu_2024}. Finally, attribution-based approaches focus on identifying the generating model or user after text generation using stylometric or statistical features rather than explicit embedded watermarks~\cite{Niess_2025,Zhao_2023}.

Although the proposed techniques differ considerably in their implementation, most existing watermarking schemes remain inherently \emph{synchronization dependent}. Watermark bits are typically embedded sequentially or organised into fixed-size algebraic blocks, so insertions, deletions, paraphrasing, or sentence reordering destroy the alignment between the embedding and detection processes. As a consequence, even small editing operations may invalidate entire watermark blocks and substantially reduce recovery performance. Developing synchronization-free watermarking schemes that remain robust under arbitrary editing therefore remains one of the major open problems in LLM watermarking.

A recent and closely related effort by Qu et al.~\cite{qu2024provably} similarly targets
efficient multi-bit watermarking through pseudo-random segment assignment combined with
Reed--Solomon error correction, achieving substantially better extraction accuracy and
speed than earlier multi-bit schemes; unlike the present work, however, their segments are
assigned via the hash of a single preceding token and recovered through per-segment
enumeration, so watermark bits sharing a segment remain positionally coupled during
extraction, whereas our construction embeds every bit as an independent algebraic
constraint recoverable in any order.

\subsection{Our Contributions}

The main contributions of this paper are summarized as follows.

\begin{itemize}
\item We propose a synchronization-free watermarking framework in which every watermark bit is embedded independently, eliminating the need for token alignment during watermark recovery.

\item We introduce a Reed--Solomon-based encoding method that reconstructs the embedded identity from any sufficiently large collection of correctly recovered congruences, irrespective of their positions in the generated text.

\item We demonstrate that all sources of corruption, including embedding errors and text editing operations such as insertions, deletions, and substitutions, can be modelled by a single Binary Symmetric Channel (BSC).

\item We derive exact and approximate probabilistic formulas that predict the minimum text length required for successful watermark recovery with a prescribed confidence level.

\item We propose lightweight variants that partition long identities into multiple independently recoverable fragments, substantially reducing computational complexity while preserving synchronization-free operation.

\item Finally, we extend the proposed watermarking framework to diffusion-based language models and present efficient token commitment algorithms together with their probabilistic analysis.
\end{itemize}

\section{Algebraic Building Blocks}
\label{sec:background}

This section introduces the algebraic framework underlying the proposed
watermarking scheme. Unlike conventional LLM watermarking techniques,
which embed watermark bits sequentially into generated tokens, our
construction associates every identity with an algebraic fingerprint
derived from a Reed--Solomon evaluation code. Each observed token pair
reveals one independently generated binary constraint on the embedded
identity. Consequently, watermark recovery does not depend on the
relative positions of the surviving tokens, making the scheme naturally
synchronization free.

Let
$
\mathcal{S}
=
\{0,1,\ldots,2^n-1\}
$
denote the space of admissible identities, where every identity is
represented by an $n$-bit binary vector
$
S=(s_0,s_1,\ldots,s_{n-1}),
s_i\in\{0,1\}.
$
Throughout the paper we assume that
$q>2^n$
is an odd prime. Every identity is uniquely represented by the polynomial
\[
f_S(x)
=
s_0+s_1x+\cdots+s_{n-1}x^{n-1}
\in
\mathbb F_q[x].
\]
Since the polynomial degree is at most $n-1$, the mapping
$S
\longleftrightarrow
f_S(x)
$
is one-to-one.
The polynomial representation provides an algebraic encoding of the
identity and allows every evaluation point
$\alpha\in\mathbb F_q$
to generate one binary observation.

For every identity $S$, define the binary evaluation function
$
\phi_S:
\mathbb F_q
\longrightarrow
\{0,1\}
$
by
\[
\phi_S(\alpha)
=
f_S(\alpha)\bmod2.
\]
The function $\phi_S$ will be referred to as the
\emph{algebraic fingerprint} of the identity $S$.
Unlike classical Reed--Solomon decoding, the receiver never observes the
field value $f_S(\alpha)$ itself. Instead, only its parity is available.
Consequently, every evaluation point produces exactly one binary
observation,
\[
w=\phi_S(\alpha),
\]
which later becomes embedded into the generated text.
The complete fingerprint of the identity is therefore the binary function
$
\Phi(S)
=
\{
(\alpha,\phi_S(\alpha))
:
\alpha\in\mathbb F_q
\}.
$
During watermark extraction only a small subset of these evaluations is
observed.

The proposed construction employs
Reed--Solomon codes. Let
$
\alpha_1,\alpha_2,\ldots,\alpha_m
\in
\mathbb F_q
$
be distinct evaluation points.
The corresponding fingerprint bits are
$
w_i
=
\phi_S(\alpha_i)
=
f_S(\alpha_i)\bmod2,
i=1,\ldots,m.
$
The vector
$
\mathbf w
=
(w_1,\ldots,w_m)
$
constitutes an algebraic fingerprint of the identity.
Observe that this differs fundamentally from classical Reed--Solomon
codes. In a standard Reed--Solomon codeword the symbols
$f_S(\alpha_i)$
are transmitted.
In contrast, our construction reveals only their parity,
\[
f_S(\alpha_i)\bmod2,
\]
which substantially reduces the amount of information revealed while
remaining sufficient for identity recovery.
The probabilistic analysis relies on the following assumption.
\vspace{3mm}

\fbox{%
\parbox{0.9\linewidth}{%
\begin{assumption}[Balanced Binary Evaluation]
\label{ass:balanced}
{\it
For every pair of distinct identities
$S,T\in\mathcal S$,
and every evaluation point chosen uniformly from
$\mathbb F_q$,
\[
\Pr
\left[
\phi_S(\alpha)
=
\phi_T(\alpha)
\right]
=
\frac12.
\]
Furthermore, evaluations performed at distinct points are assumed to be
statistically independent.
}
\end{assumption}
}}

\vspace{3mm}
Assumption~\ref{ass:balanced} reflects the behaviour observed
experimentally for randomly selected evaluation points over sufficiently
large finite fields. It is analogous to the random oracle assumption
commonly adopted in cryptographic security analyses and allows the
fingerprint to be analysed as a sequence of independent binary tests.

\begin{remark}[Scope of the independence assumption]
Assumption~\ref{ass:balanced} concerns evaluations at \emph{distinct} points
$\alpha_i\ne\alpha_j$. In the construction of Section~\ref{sec:construction}, the
evaluation point is derived from the identity of a single preceding token,
$\alpha_i = H(K,\mathrm{id}(t_i))\bmod q$. Consequently, two token pairs sharing the
same preceding token necessarily yield the same evaluation point and therefore the
same fingerprint bit; the independence assumption applies only across pairs whose
preceding tokens are distinct, and repeated common tokens (e.g.\ function words)
contribute no additional independent observation beyond the first occurrence. Should
this reduction in effective sample size be a concern for a given text distribution,
the evaluation point can instead be derived from a short context window, for
instance $\alpha_i = H(K,\mathrm{id}(t_{i-1}),\mathrm{id}(t_i))\bmod q$, so that
repetitions of a single token no longer collide unless the preceding context also
repeats, at the cost of a (typically mild) increase in synchronisation sensitivity.
A related tension appears in Qu et al.~\cite{qu2024provably}, where a token's
assigned message segment is likewise selected via the hash of the single preceding
token; the resulting imbalance in segment assignment (frequent tokens routing
disproportionately many observations to the same segment) is precisely the same
phenomenon, which they address with a dynamic-programming-based balanced assignment
over the vocabulary rather than by widening the hash context.
\end{remark}

\subsection{Probabilistic Analysis}

Each observed fingerprint bit partitions the identity space into two
subsets according to the value of
$\phi_S(\alpha).$
Under Assumption~\ref{ass:balanced}, each observation removes,
on average, one half of all remaining candidate identities.
\begin{lemma}
\label{lem:candidates}
Suppose that the observed fingerprint consists of
$m$
independent binary evaluations.
Then the expected number of identities consistent with all observations
equals
$2^{\,n-m}.$
\end{lemma}
\begin{proof}
Initially the identity space contains
$2^n$
candidates.
Every independent binary evaluation eliminates one half of the remaining
candidates in expectation.
Consequently, after observing
$m$
evaluations,
$2^n
\left(\frac12\right)^m
=
2^{\,n-m}$
candidate identities remain on average.
\hfill$\Box$
\end{proof}
The previous lemma shows that every binary evaluation approximately
halves the number of candidate identities. We next derive a sufficient
condition for unique identification of the embedded identity.
\begin{theorem}
\label{thm:main}
Suppose that the binary fingerprint evaluations satisfy
Assumption~\ref{ass:balanced}. Let
$m=n+\lambda,
\lambda\ge0,$
be the number of independently observed fingerprint bits.
Then the probability that the embedded identity cannot be recovered is
bounded by
$\Pr[\mathrm{Failure}]
<
2^{-\lambda}.$
\end{theorem}
\begin{proof}
Fix the embedded identity $S$ and let $\alpha_1,\ldots,\alpha_m$ denote the observed
evaluation points, with $w_i=\varphi_S(\alpha_i)$ for $i=1,\ldots,m$. Recovery fails
whenever some identity $T\ne S$ agrees with $S$ on all $m$ observed points, that is,
$\varphi_T(\alpha_i)=\varphi_S(\alpha_i)$ for every $i=1,\ldots,m$; note this is a much
weaker event than $T$ and $S$ agreeing as functions on all of $\mathbb F_q$, since we
only require agreement at the $m$ observed points.

Fix any $T\ne S$. By Assumption~\ref{ass:balanced}, $\Pr[\varphi_T(\alpha_i)=\varphi_S(\alpha_i)]=\frac12$
for every evaluation point $\alpha_i$, and since the observations are independent,
\[
\Pr\bigl[\varphi_T(\alpha_i)=\varphi_S(\alpha_i)\ \text{for all}\ i=1,\ldots,m\bigr] = 2^{-m}.
\]
There are at most $2^n-1$ identities $T\ne S$. Applying the union bound over all such $T$,
\[
\Pr[\mathrm{Failure}] \;\le\; (2^n-1)\,2^{-m} \;<\; 2^{n-m}.
\]
Substituting $m=n+\lambda$ yields $\Pr[\mathrm{Failure}]<2^{-\lambda}$, which completes the
proof.
\hfill$\Box$
\end{proof}
The watermark construction requires a collection of evaluation points
$\alpha_1,\ldots,\alpha_m
\in
\mathbb F_q.$
Since every evaluation produces an independent binary observation,
the evaluation points should be distinct.

Throughout the paper they are generated pseudorandomly from a secret key
using a cryptographic hash function
\[
\alpha_i
=
H(K,\tau_i)\bmod q,
\]
where $\tau_i$ denotes the identifier associated with the current token
pair and $K$ is known only to the watermark detector.
The probability that two independently generated evaluation points
coincide follows directly from the birthday paradox.
\begin{proposition}
Suppose $m$ evaluation points are sampled independently and uniformly from
$\mathbb F_q$. Then
\[
\Pr[\mathrm{collision}]
<
\frac{m(m-1)}{2q}.
\]
\end{proposition}
\begin{proof}
For every pair of sampled evaluation points,
$\Pr[\alpha_i=\alpha_j]
=
\frac1q.$
Applying the union bound over the $\binom{m}{2}$ pairs gives
\[
\Pr[\mathrm{collision}]
\le
\binom{m}{2}\frac1q
<
\frac{m(m-1)}{2q}.
\]
\hfill$\Box$
\end{proof}

The preceding analysis establishes that the proposed fingerprinting
mechanism possesses two desirable properties.
First, every observed token pair contributes one statistically
independent binary constraint on the embedded identity.
Consequently, approximately one additional observation eliminates one
additional bit of uncertainty.
Second, the recovery probability depends only on the number of surviving
observations and not on their positions within the generated text.
This property distinguishes the proposed construction from existing
block-based watermarking schemes, whose recovery depends on maintaining
synchronization between the embedding and extraction processes.
The algebraic fingerprint introduced in this section forms the
mathematical foundation of the watermarking scheme described in the next
section.

\section{Watermark Construction}
\label{sec:construction}
The algebraic framework developed in the previous section associates
every identity with a binary fingerprint function
$\phi_S:\mathbb F_q\rightarrow\{0,1\}.$
The watermarking scheme embeds evaluations of this function into the
generated text. Unlike existing synchronization-dependent watermarking
schemes, each embedded watermark bit is completely independent of all
previously generated bits. Consequently, watermark recovery depends only
on the number of surviving observations rather than their positions in
the text.

The construction uses the following public parameters.
\begin{itemize}
\item
An odd prime
$q>2^n.$
\item
A cryptographic hash function
$H:\{0,1\}^*\rightarrow\mathbb F_q.$
\item
A language model capable of generating two disjoint token classes
corresponding to binary values $0$ and $1$.
\end{itemize}
The detector additionally possesses the secret key $K$,
which determines the sequence of evaluation points.

Let $S=(s_0,\ldots,s_{n-1})$ be the identity to be embedded.
Construct the polynomial
\[
f_S(x)
=
s_0+s_1x+\cdots+s_{n-1}x^{n-1}
\in\mathbb F_q[x].
\]
The associated binary fingerprint is
$\phi_S(\alpha)
=
f_S(\alpha)\bmod2.$
No further preprocessing of the identity is required.

Watermark embedding proceeds independently for every consecutive token
pair $P_i=(t_i,t_{i+1}).$ The first token determines the evaluation point, while the second token
carries the watermark bit.
For every pair $P_i$, the sender performs the following operations.
\begin{enumerate}
\item
Compute
$\alpha_i
=
H(K,\mathrm{id}(t_i))
\bmod q.$
\item
Evaluate the fingerprint
$w_i
=
\phi_S(\alpha_i).$
\item
Generate the second token $t_{i+1}$
using any binary watermarking mechanism that biases generation toward
the token class corresponding to $w_i$.
\end{enumerate}
Since every evaluation point depends only on the current token pair,
each watermark bit is generated independently of all previous
observations.


Suppose that, after arbitrary editing, the detector identifies
$r$ surviving token pairs
$P_{i_1},
\ldots,
P_{i_r}.$
For every surviving pair the detector recomputes
$\alpha_j
=
H(K,\mathrm{id}(t_{i_j}))
\bmod q $
and determines the embedded bit
$\hat w_j$
from the token class of the second token.
This produces the observation set
\[
\{
(\alpha_j,\hat w_j)
\}_{j=1}^{r}.
\]
Unlike conventional watermarking schemes, these observations may appear
in arbitrary order.
No synchronization between embedding and extraction is required.

The detector reconstructs the embedded identity by searching for the
unique polynomial $f_S(x)$
whose binary fingerprint agrees with all observed evaluations,
\[
\phi_S(\alpha_j)
=
\hat w_j,
\qquad
j=1,\ldots,r.
\]
Recovery succeeds whenever a unique identity satisfies the observed
constraints.
According to Theorem~\ref{thm:main}, if
$r
\ge
n+\lambda,$
then 
$\Pr[\mathrm{Failure}]
<
2^{-\lambda}.$
Therefore successful watermark recovery depends only on the number of
surviving token pairs and is independent of their original positions in
the generated text.


The proposed construction differs fundamentally from existing
block-based watermarking techniques.
In conventional watermarking schemes, watermark bits are embedded
sequentially.
Insertions or deletions shift subsequent watermark positions,
destroying synchronization between the encoder and decoder.
In contrast, every watermark bit in the proposed construction is bound
only to its own token pair through the cryptographic hash function.
Consequently,
\begin{itemize}

\item insertions create additional token pairs without affecting
previous observations;
\item deletions merely remove individual observations;
\item substitutions affect only the modified token pair;
\item sentence reordering leaves every surviving observation unchanged.
\end{itemize}
The detector therefore operates on an unordered collection of binary
evaluations rather than on a synchronized bit stream.
This synchronization-free property constitutes the principal advantage
of the proposed watermarking framework.

\section{Recovery Algorithms}
\label{sec:recovery}
During watermark extraction the detector observes a collection
$\mathcal O=
\{(\alpha_i,\hat w_i)\}_{i=1}^{r},$
where
$\hat w_i=\phi_S(\alpha_i)
=f_S(\alpha_i)\bmod2$
denotes the recovered fingerprint bit associated with the evaluation
point $\alpha_i$.
The objective of the recovery algorithm is to determine the unique
identity
$S\in\mathcal S$
whose fingerprint agrees with all observed evaluations. Formally, the
detector searches for an identity satisfying
$\phi_S(\alpha_i)=\hat w_i,
i=1,\ldots,r.$
According to Theorem~\ref{thm:main}, if
$r\ge n+\lambda,$
the solution is unique except with probability at most
$2^{-\lambda}$.
Depending on the size of the identity space, different recovery
algorithms become appropriate.
\vspace{3mm}

\noindent
{\bf Algorithm 1: Exhaustive Search}
%
When the identity space is sufficiently small, recovery can be
performed by exhaustive enumeration.
For every candidate identity
$S'
\in
\mathcal S,$
the detector constructs the corresponding polynomial
$f_{S'}(x)$ and evaluates
\[
\phi_{S'}(\alpha_i)
=
f_{S'}(\alpha_i)\bmod2,
\qquad
i=1,\ldots,r.
\]
The resulting fingerprint is compared with the observed watermark bits.
The unique candidate satisfying all observations is returned as the
embedded identity.
\begin{algorithm}[H]
\caption{Exhaustive Recovery}
\label{alg:lookup}
\begin{algorithmic}[1]
\Require Observations
$\mathcal O=\{(\alpha_i,\hat w_i)\}_{i=1}^{r}$

\ForAll{$S'\in\mathcal S$}

    \State Construct $f_{S'}(x)$.

    \State Compute
    $\phi_{S'}(\alpha_i)$
    for every observation.

    \If{all fingerprint bits agree}

        \State Return $S'$.

    \EndIf

\EndFor

\State Return {\sc Fail}.

\end{algorithmic}
\end{algorithm}
\noindent
The complexity of exhaustive recovery is $O(2^n),$
making it practical whenever $n\lesssim20.$
\vspace{3mm}

\noindent
{\bf Algorithm 2: Meet-in-the-Middle Recovery}
For larger identities the exhaustive search becomes infeasible.
A significant reduction in complexity is obtained by partitioning the
identity into halves.
Let $n=n_L+n_R,$ 
and write
$S=(S_L,S_R),$
where
\[
S_L
=
(s_0,\ldots,s_{n_L-1}),
\qquad
S_R
=
(s_{n_L},\ldots,s_{n-1}).
\]
The polynomial representation decomposes naturally as
$f_S(x)
=
f_L(x)
+
x^{n_L}f_R(x).$
The algorithm first enumerates every possible lower half
$S_L,$
computes its contribution to every observed evaluation point, and stores
the resulting vectors in a lookup table.
The second stage enumerates every upper half
$S_R$ and searches for compatible entries in the lookup table.
The search therefore requires approximately
$2^{n/2}$ operations instead of $2^n$.
\begin{algorithm}[H]
\caption{Meet-in-the-Middle Recovery}\label{alg:mitm}
\begin{algorithmic}[1]

\Require
Observed fingerprint
$\mathcal O$

\State Construct a lookup table containing the evaluation vectors of
all lower-half identities.

\ForAll{upper-half identities}

    \State Compute the corresponding evaluation vector.

    \State Search the lookup table for a compatible lower-half.

    \If{a unique match exists}

        \State Return the reconstructed identity.

    \EndIf

\EndFor

\State Return {\sc Fail}.

\end{algorithmic}
\end{algorithm}
Ignoring the cost of hash-table operations, the computational
complexity becomes
$O(2^{n/2}),$
while the memory requirement is
$O(2^{n/2}).$
Consequently, the meet-in-the-middle algorithm extends practical
recovery to substantially larger identity spaces.

{\small
\begin{example}
We illustrate the meet-in-middle recovery algorithm
(Algorithm~\ref{alg:mitm}) with parameters $n=4$, $q=17$, split
position $k=2$, true identity $S=6=(s_0,s_1,s_2,s_3)=(0,1,1,0)$,
so that
$f_6(x) = x + x^2 \in \mathbb F_{17}[x],$
and evaluation points $\alpha_1=2$, $\alpha_2=3$, $\alpha_3=5$,
$\alpha_4=7$.

\noindent
\underline{\it Observed bits.}
Each observed bit is $w_i = \bigl(f_6(\alpha_i) \bmod 17\bigr) \bmod 2$.
\begin{center}
\begin{tabular}{cccc}
\hline
$i$ & $\alpha_i$ & $f_6(\alpha_i) \bmod 17$ & $w_i$ \\
\hline
$1$ & $2$ & $6$  & $0$ \\
$2$ & $3$ & $12$ & $0$ \\
$3$ & $5$ & $13$ & $1$ \\
$4$ & $7$ & $5$  & $1$ \\
\hline
\end{tabular}
\end{center}
This gives the observed bit vector $\hat{\mathbf w} = (0,0,1,1)$.

\noindent
\underline{\it Splitting the identity.}
Write the left and right halves of $S=(s_0,s_1,s_2,s_3)$ as the
integers
\[
S_L = s_0 + 2 s_1 = 2,
\qquad
S_R = s_2 + 2 s_3 = 1,
\]
so that the full identity is recovered as $S = S_L + 2^{k} S_R$.
Correspondingly, $f_S(x)$ splits as
\[
f_S(x) = \underbrace{(s_0 + s_1 x)}_{h_{S_L}(x)}
       \;+\; x^{k}\underbrace{(s_2 + s_3 x)}_{\textstyle =: g_{S_R}(x)},
\]
so $h_{S_L}$ depends only on the left half $(s_0,s_1)$ and
$g_{S_R}$ only on the right half $(s_2,s_3)$. For the true identity,
$h_2(x) = x$ and $g_1(x) = x^2$.

\noindent
\underline{\it Baby step: precompute the left-half table.}
For every candidate $S_L \in \{0,1,2,3\}$, precompute
$\mathbf r_L(S_L) =
\bigl(h_{S_L}(\alpha_i) \bmod q\bigr)_{i=1}^{4}
\in \{0,\ldots,q-1\}^{4}.$
This table depends only on $S_L$ and is built once, independently
of $S_R$.
\begin{center}
\begin{tabular}{cccc}
\hline
$S_L$ & $(s_0,s_1)$ & $h_{S_L}(x)$ & $\mathbf r_L(S_L)$ \\
\hline
$0$ & $(0,0)$ & $0$   & $(0,0,0,0)$ \\
$1$ & $(1,0)$ & $1$   & $(1,1,1,1)$ \\
$2$ & $(0,1)$ & $x$   & $(2,3,5,7)$ \\
$3$ & $(1,1)$ & $1+x$ & $(3,4,6,8)$ \\
\hline
\end{tabular}
\end{center}

\noindent
\underline{\it Giant step: search over the right half.}
For every candidate $S_R \in \{0,1,2,3\}$, compute
$\mathbf r_R(S_R) =
\bigl(g_{S_R}(\alpha_i) \bmod q\bigr)_{i=1}^{4}
\in \{0,\ldots,q-1\}^{4},$
then scan the baby-step table for an $S_L$ such that
\[
\bigl(r_L^{(i)}(S_L) + r_R^{(i)}(S_R)\bigr) \bmod q
\ \equiv\ w_i \pmod 2
\qquad\text{for every } i=1,\ldots,4.
\]
In practice the constraints are checked in order
$i=1,2,3,4$, and a candidate $S_L$ is discarded as soon as a
single constraint fails; this early termination is what makes the
search efficient in general.

\begin{center}
\renewcommand{\arraystretch}{1.2}
\begin{tabular}{cccc}
\hline
$S_R$ & $(s_2,s_3)$ & $\mathbf r_R(S_R)$ & Result \\
\hline
$0$ & $(0,0)$ & $(0,0,0,0)$  & no $S_L$ satisfies all four constraints \\
$1$ & $(1,0)$ & $(4,9,8,15)$ & $S_L = 2$ satisfies all four constraints \checkmark \\
$2$ & $(0,1)$ & $(8,10,6,3)$ & no $S_L$ satisfies all four constraints \\
$3$ & $(1,1)$ & $(12,2,14,1)$& no $S_L$ satisfies all four constraints \\
\hline
\end{tabular}
\end{center}

\noindent
\underline{\it Verifying the unique match $S_L=2$, $S_R=1$.}
\begin{center}
\renewcommand{\arraystretch}{1.2}
\begin{tabular}{ccccccc}
\hline
$i$ & $r_L^{(i)}(2)$ & $r_R^{(i)}(1)$ & sum & sum $\bmod 17$ & parity & $w_i$ \\
\hline
$1$ & $2$ & $4$  & $6$  & $6$  & $0$ & $0$ \checkmark \\
$2$ & $3$ & $9$  & $12$ & $12$ & $0$ & $0$ \checkmark \\
$3$ & $5$ & $8$  & $13$ & $13$ & $1$ & $1$ \checkmark \\
$4$ & $7$ & $15$ & $22$ & $5$  & $1$ & $1$ \checkmark \\
\hline
\end{tabular}
\end{center}
Note the wraparound at $i=4$: although $r_L^{(4)}=7$ and
$r_R^{(4)}=15$ are individually odd, their sum $22$ exceeds $q=17$,
so $22 \bmod 17 = 5$ is \emph{odd}, not even. Combining the two
parities directly (e.g.\ via XOR) before reducing modulo $q$ would
therefore give the wrong bit; the reduction mod $q$ must be applied
\emph{before} taking parity.
Since $(S_L,S_R)=(2,1)$ is the unique pair satisfying all four
constraints, the identity is recovered as
$\hat S = S_L + 2^{k} S_R = 2 + 4\cdot 1 = 6,$
correctly matching the true identity $S=6$. \checkmark
\end{example}
}
\begin{table}[ht]
\centering
\caption{Comparison of recovery algorithms.}
\label{tab:recovery-comparison}
\begin{tabular}{lccc}
\hline
Algorithm & Time & Memory & Typical range \\
\hline
Exhaustive search      & $O(2^n)$     & $O(1)$        & $n\le20$ \\
Meet-in-the-middle     & $O(2^{n/2})$ & $O(2^{n/2})$  & $20<n\le60$ \\
\hline
\end{tabular}
\end{table}

For $n=60$ the baby-step table holds $2^{30}\approx1.07\times10^9$ entries, requiring
roughly $9$--$17$ GB of memory depending on encoding --- comfortably within reach of a
single well-specified workstation. The method remains correct up to $n\approx64$, but
at that point the table grows to $2^{32}\approx4.3\times10^9$ entries
($\sim35$--$70$ GB), which requires a dedicated high-memory machine rather than typical
hardware; we therefore report $n\le60$ as the practical range and treat larger $n$ up to
$64$ as an achievable but hardware-demanding edge case.

For substantially larger identity spaces, the recovery problem may also
be formulated as a system of Boolean constraints or as an instance of
the Closest Vector Problem in an appropriately constructed lattice.
Although these formulations appear promising, a detailed complexity
analysis is beyond the scope of the present work and is therefore left
for future research.

\section{Security Analysis}
\label{sec:security}

The security of the proposed watermarking scheme relies on two
independent principles. First, the embedded identity must be recovered
uniquely from the observed watermark. Second, the watermark must remain
recoverable after the generated text has undergone editing operations
such as insertions, deletions or substitutions.
Unlike existing synchronization-dependent watermarking schemes, the
proposed construction embeds each fingerprint bit independently.
Consequently, the combined effect of all error mechanisms can be
modelled by a Binary Symmetric Channel (BSC), which enables a simple
probabilistic analysis of watermark recovery.


Suppose that the detector observes $r$ watermark bits generated from an
embedded identity $S\in\mathcal S$.
According to Theorem~\ref{thm:main}, if $r\ge n+\lambda$, the probability
that another identity satisfies all observed fingerprint constraints is
at most $2^{-\lambda}$. Consequently, the probability of falsely
attributing a document to another registered identity decreases
exponentially with the number of surviving watermark observations.


The proposed watermarking scheme assumes that every token pair embeds one
binary fingerprint bit
$w=\phi_S(\alpha)=f_S(\alpha)\bmod2,$
where the evaluation point $\alpha$ is derived from the preceding token
using the secret-keyed hash function.
During embedding and subsequent editing of the generated text, individual
fingerprint bits may become corrupted. We model the combined effect of
all error mechanisms by a Binary Symmetric Channel (BSC) with crossover
probability $p$. Each watermark bit is therefore recovered correctly
with probability $1-p$ and incorrectly with probability $p$,
independently of all remaining bits.
This model is particularly well suited to the proposed construction
because every fingerprint observation is statistically independent of
every other observation.

\subsection{Why the BSC Model Applies}

Most existing LLM watermarking schemes embed watermark bits in
synchronized blocks. A single insertion or deletion usually destroys the
alignment of the remaining watermark bits, causing an entire block to be
decoded incorrectly.
The proposed scheme behaves fundamentally differently. Every watermark
bit depends only on one token pair through the evaluation point
$\alpha=H(K,\mathrm{id}(t))$. Consequently,
\begin{itemize}
\item each token pair contributes one independent fingerprint bit;
\item editing operations affect only local observations;
\item watermark recovery depends only on the number of surviving
observations and not on their positions within the document.
\end{itemize}
Therefore all editing operations may be represented by a single channel
parameter $p$ rather than analysed individually.
The effective crossover probability $p$ reflects both imperfections of
the language model and adversarial editing of the generated text.

\paragraph{Embedding errors.}
Despite applying a logit bias, the language model may occasionally
generate a token from the wrong colour class, producing one erroneous
fingerprint bit.
\paragraph{Insertion.}
Suppose an adversary inserts a token $t_f$ between two consecutive tokens
$t_i$ and $t_{i+1}$, replacing one watermarked pair by the two pairs
$(t_i,t_f)$ and $(t_f,t_{i+1})$. For the insertion to remain completely
undetected, two independent events must occur simultaneously: (i) the
inserted token must belong to the required colour class, and (ii) its
hash-derived evaluation point must produce the expected fingerprint bit
for the subsequent token. Each event occurs with probability
approximately $1/2$, yielding
$\Pr[\text{undetected insertion}]
\approx
\frac14.$
Hence an insertion is detected with probability approximately
$3/4$. Equivalently, among $u$ random insertions, only about $u/4$
remain invisible to the detector, while approximately $3u/4$ introduce
at least one detectable watermark inconsistency.

\paragraph{Deletion.}

Suppose an adversary deletes the token $t_{i+1}$ from the sequence
$(t_i,t_{i+1},t_{i+2})$. The original watermark observations associated
with the pairs $(t_i,t_{i+1})$ and $(t_{i+1},t_{i+2})$ disappear and are
replaced by the single pair $(t_i,t_{i+2})$.
For the deletion to remain completely undetected, the newly formed pair
must simultaneously satisfy two independent conditions. First,
$t_{i+2}$ must belong to the colour class prescribed by the watermark
bit associated with the new evaluation point generated from $t_i$.
Second, the newly generated evaluation point must produce the correct
fingerprint relation for the surviving token pair. Under the random
oracle assumption, each event occurs with probability approximately
$1/2$. Hence
$\Pr[\text{undetected deletion}]
\approx
\frac14,$
and therefore
$\Pr[\text{detected deletion}]
\approx
\frac34.$
Consequently, among $u$ random deletions, approximately $u/4$ remain
undetected, whereas about $3u/4$ introduce at least one detectable
watermark inconsistency.


\paragraph{Substitution.}

Suppose an adversary replaces a token $t_i$ by another token $t'_i$.
This modification changes both the colour class of the substituted token
and the evaluation point used to generate the subsequent fingerprint
constraint.
For the substitution to remain undetected, two independent events must
occur simultaneously. The replacement token must belong to the expected
colour class, and its hash-derived evaluation point must generate the
correct fingerprint relation for the following token. Assuming both
events occur independently with probability approximately $1/2$, we
obtain
$\Pr[\text{undetected substitution}]
\approx
\frac14,$
and therefore
$\Pr[\text{detected substitution}]
\approx
\frac34.$
Thus, among $u$ random substitutions, approximately $u/4$ remain
undetected, while about $3u/4$ produce a detectable violation of the
embedded watermark.

Paraphrasing, deletion of contiguous text fragments, and insertion of multiple consecutive tokens can all be decomposed into sequences of the elementary editing operations analysed above. Consequently, their effect on the watermark can be modelled by composing the corresponding insertion, deletion, and substitution probabilities, yielding an effective crossover probability for the resulting Binary Symmetric Channel.

\subsection{Probability of Successful Recovery}

Suppose that the watermark is embedded into $N$ token pairs and let
$X\sim\mathrm{Binomial}(N,1-p)$
denote the number of correctly recovered fingerprint bits.
Since the proposed Reed--Solomon construction reconstructs the embedded
identity from any collection of at least $n$ correct observations,
successful recovery occurs whenever $X\ge n$. Thus,
\[
\Pr[\mathrm{Recovery}]
=
\Pr[X\ge n].
\]
For a desired recovery probability $\gamma$, the minimum watermark length
is therefore the smallest integer $N$ satisfying
$\Pr[X\ge n]\ge\gamma.$
For moderate values of $N$, the binomial distribution may be accurately
approximated by a normal distribution, yielding
\[
N^*
\approx
\left\lceil
\frac{1}{1-p}
\left(
\frac{
z_\gamma\sqrt{p(1-p)}
+
\sqrt{
z_\gamma^2p(1-p)+4n(1-p)
}
}{2}
\right)^2
\right\rceil,
\]
where $z_\gamma$ denotes the standard normal quantile associated with
probability $\gamma$.

{\small
\begin{example}

For a $32$-bit identity and target recovery probability
$\gamma=0.99$, Table~\ref{tab:bsc} compares the exact binomial solution
with the normal approximation.
\begin{table}[ht]
\centering
\caption{Minimum number of transmitted token pairs required for reliable
recovery of a $32$-bit identity.}
\label{tab:bsc}
\begin{tabular}{ccccc}
\hline
$p$ & Exact $N^*$ & Approx. $N^*$ & Overhead & Success probability\\
&&&$N^*-n$&   \\
\hline
0.001 & 33 & 33 & 1  & 0.9995\\
0.005 & 34 & 34 & 2  & 0.9993\\
0.010 & 34 & 34 & 2  & 0.9953\\
0.050 & 37 & 37 & 5  & 0.9905\\
0.100 & 41 & 41 & 9  & 0.9939\\
0.200 & 48 & 49 &16  & 0.9907\\
0.300 & 57 & 58 &25  & 0.9908\\
\hline
\end{tabular}
\end{table}
The results demonstrate that the redundancy required by the proposed
scheme is remarkably small. Even for an extremely noisy channel with
crossover probability $p=0.30$, only $25$ additional token pairs are
required to recover a $32$-bit identity with probability at least
$99\%$. Under more realistic editing conditions ($p \leq 0.10$), fewer than
ten additional token pairs are sufficient. Consequently, reliable
watermark recovery can typically be achieved from only a few sentences of
generated text.
\end{example}
}

The Binary Symmetric Channel representation constitutes one of the main
advantages of the proposed watermarking framework. Because every
fingerprint bit is generated independently, synchronization errors never
propagate through the watermark. Instead, insertions, deletions,
substitutions and embedding imperfections merely increase the effective
crossover probability $p$.
Consequently, the complete robustness analysis reduces to a classical
coding-theoretic problem whose behaviour is fully characterised by the
binomial distribution. This considerably simplifies both the analysis
and the practical design of the watermarking system while providing
explicit guarantees on the probability of successful watermark recovery.

\section{Scaling to Long Identities by Fragmentation}
\label{sec:fragmentation}

The recovery algorithms presented in Section~\ref{sec:recovery} become
progressively more expensive as the identity length increases.
To overcome this limitation, we partition the secret identity into
several shorter fragments and recover each fragment independently.
Consequently, one large recovery problem is replaced by a collection of
small independent recovery problems, while preserving all desirable
properties of the proposed watermarking scheme, including
synchronization-free embedding and robustness against editing.
Suppose that the secret identity consists of $n=mv$ bits and is divided
into $v=2^r$ fragments
$S^{(1)},S^{(2)},\ldots,S^{(v)}$, each containing $m$ bits. Every
fragment is encoded independently using the Reed--Solomon construction
described in Section~\ref{sec:construction}.


For every token pair $(t_i,t_{i+1})$, the keyed hash
$H(K,\mathrm{id}(t_i))$ is interpreted as two independent random
variables. The first $r=\log_2v$ bits determine the fragment index
$j(i)\in\{1,\ldots,v\}$, while the remaining bits generate the
evaluation point $\alpha_i\in\mathbb F_q$ for that fragment. The
embedded watermark bit is therefore
$
w_i=f_{S^{(j(i))}}(\alpha_i)\bmod2,
$
where $f_{S^{(j)}}$ denotes the Reed--Solomon polynomial representing
fragment $S^{(j)}$.
Since both the fragment index and the evaluation point are generated
independently from the keyed hash function, every watermark observation
remains independent of all other observations. Consequently, the
fragmented construction remains synchronization-free.

\paragraph{Minimum Transmission Length}
%
Assume that each embedded watermark bit is transmitted through a Binary
Symmetric Channel with crossover probability $p$.
For a fixed fragment, each transmitted token pair contributes one
correct watermark observation with probability
$
q=\frac{1-p}{v},
$
since the corresponding fragment is selected with probability $1/v$ and
its watermark bit is recovered correctly with probability $1-p$.

Let $X_j$ denote the number of correctly recovered observations for the
$j$-th fragment. The random variables
$X_1,\ldots,X_v$ follow a multinomial distribution. Since $q$ is small,
we approximate each $X_j$ by an independent random variable
$
X_j\sim\mathrm{Binomial}(N,q).
$
Recovery of fragment $j$ succeeds whenever $X_j\ge m$.
Consequently, successful recovery of the complete identity requires
$
X_j\ge m,  j=1,\ldots,v.
$
Using the independence approximation,
$
\Pr[\text{Recovery}]
\approx
\Pr[X_1\ge m]^v.
$
For a prescribed recovery probability $\gamma$, the required
per-fragment success probability is therefore
$
\gamma_0=\gamma^{1/v},
$
and the minimum transmission length $N^*$ is the smallest integer
satisfying
$
\Pr[\mathrm{Binomial}(N,q)\ge m]\ge\gamma_0.
$
Since $q=(1-p)/v$ is typically small, the Binomial distribution is well
approximated by a Poisson distribution with mean
$
\lambda=Nq.
$
Let $\lambda^*$ denote the smallest value satisfying
\[
1-F_{\mathrm{Pois}(\lambda)}(m-1)\ge\gamma_0,
\]
where $F_{\mathrm{Pois}(\lambda)}$ is the cumulative Poisson
distribution.
The required transmission length is therefore approximated by
$
N^*
\approx
\left\lceil
\frac{v\lambda^*}{1-p}
\right\rceil.
$
This expression shows that the required text length grows linearly with
the number of fragments and inversely with the channel reliability
$1-p$.

{\small
\begin{example}
Table~\ref{tab:multipart} considers a $128$-bit identity divided into
$v=32$ fragments of $m=4$ bits each. The target recovery probability is
$\gamma=0.99$, giving
$\gamma_0\approx0.999686$ and
$\lambda^*\approx14.52$.
Table~\ref{tab:multipart-64} considers a $64$-bit identity divided into
$v=8$ fragments of $m=8$ bits each for target recovery probabilities
$\gamma=0.99$ and $\gamma=0.90$.
\begin{table}[ht]
\centering
\caption{Minimum transmission length $N^*$ for $n=128$, $m=4$, $v=32$,
$\gamma=0.99$.}
\label{tab:multipart}
\renewcommand{\arraystretch}{1.15}
\begin{tabular}{cccccc}
\hline
$p$ & $q=(1-p)/v$ & Exact $N^*$ & Poisson $N^*$ & Overhead & Success \\
&&&&$N^*-vm$& probability\\
\hline
0.001 & 0.031219 & 460 & 466 & 332 & 0.9902\\
0.005 & 0.031094 & 461 & 467 & 333 & 0.9900\\
0.010 & 0.030937 & 464 & 470 & 336 & 0.9902\\
0.050 & 0.029687 & 484 & 490 & 356 & 0.9902\\
0.100 & 0.028125 & 511 & 517 & 383 & 0.9902\\
0.200 & 0.025000 & 575 & 581 & 447 & 0.9901\\
0.300 & 0.021875 & 658 & 664 & 530 & 0.9901\\
\hline
\end{tabular}
\end{table}
\begin{table}[ht]
\centering
\caption{Minimum transmission length $N^*$ for $n=64$, $m=8$, $v=8$.}
\label{tab:multipart-64}
\renewcommand{\arraystretch}{1.15}
\begin{tabular}{ccrrrrrr}
\hline
&&\multicolumn{3}{c}{$\gamma=0.99$}&
\multicolumn{3}{c}{$\gamma=0.90$}\\
\cline{3-5}\cline{6-8}
$p$ & $q$
& Exact
& Poisson
& Overhead
& Exact
& Poisson
& Overhead\\
&&&&$N^*-vm$&&&$N^*-vm$\\
\hline
0.001 & 0.12488 & 149 & 155 & 85 & 121 & 125 & 57\\
0.005 & 0.12438 & 149 & 156 & 85 & 121 & 126 & 57\\
0.010 & 0.12375 & 150 & 156 & 86 & 122 & 126 & 58\\
0.050 & 0.11875 & 157 & 163 & 93 & 127 & 131 & 63\\
0.100 & 0.11250 & 166 & 172 & 102 & 134 & 139 & 70\\
0.200 & 0.10000 & 187 & 193 & 123 & 152 & 156 & 88\\
0.300 & 0.08750 & 215 & 221 & 151 & 174 & 178 & 110\\
\hline
\end{tabular}
\end{table}
\end{example}
}

The results demonstrate that fragmentation enables watermarking of
arbitrarily long identities while preserving synchronization-free
embedding. As predicted by
$
N^*\approx v\lambda^*/(1-p),
$
the required transmission length grows approximately linearly with the
number of fragments and inversely with the probability of correct
watermark recovery.
For small fragment sizes, the Poisson approximation closely matches the
exact Binomial solution over the entire range of channel error
probabilities considered. In particular, for $m=4$ the approximation
differs from the exact solution by at most six transmitted token pairs,
while for $m=8$ the difference never exceeds seven token pairs.
Consequently, the Poisson approximation provides a simple and accurate
design rule for determining the minimum document length required for
reliable watermark recovery.

\begin{remark}
The independence approximation
$
\Pr[\min_jX_j\ge m]\approx\Pr[X_1\ge m]^v
$
is slightly conservative because the random variables
$X_1,\ldots,X_v$ are negatively correlated through the fixed total
number of transmitted token pairs. Consequently, the true probability of
successful recovery is typically slightly larger than that reported in
Tables~\ref{tab:multipart} and
\ref{tab:multipart-64}.
\end{remark}

\section{Evaluation Methodology}
\label{sec:evaluation}

Evaluation of LLM watermarking schemes is traditionally based on
large-scale experiments involving several language models, benchmark
datasets and collections of editing attacks. Such experiments are
important because they demonstrate that a watermarking method is
implementable and effective under practical conditions. However, purely
experimental evaluation also suffers from several fundamental
limitations.

First, the outcome inevitably depends on the particular LLMs selected for
the experiments. Since new language models and improved versions are
released continuously, it is impossible to evaluate every relevant model.
Consequently, experimental results often become outdated within a short
period of time.

Second, the performance of a watermarking scheme depends strongly on the
underlying text distribution. Scientific papers, source code, news
articles, legal documents and conversational text exhibit very different
statistical properties, resulting in different watermarking behaviour.
Any experimental study therefore reflects only the particular corpora that
were chosen.

Third, the robustness of a watermark depends on numerous implementation
parameters, including the logit bias used during embedding, the decoding
strategy of the language model, and the types of adversarial editing
performed after generation. Exploring all combinations of these
parameters rapidly becomes computationally infeasible.

For these reasons, we argue that experimental evaluation should be
complemented by mathematical modelling. Instead of analysing each LLM,
each corpus and each attack separately, it is desirable to replace their
combined effect by a statistical communication channel whose parameters
capture the overall probability of watermark corruption.

In the proposed scheme this abstraction is particularly natural.
Each embedded watermark bit is recovered independently of all other
watermark bits, and every source of corruption simply changes the value of
that bit with some probability. Consequently, the complete watermarking
process can be modelled by a Binary Symmetric Channel (BSC) with crossover
probability $p$. The parameter $p$ incorporates the combined influence of

\begin{itemize}
\item imperfect watermark embedding caused by the language model,
\item the choice of logit bias and decoding algorithm,
\item statistical properties of the generated text,
\item token insertions, deletions and substitutions,
\item paraphrasing and other editing operations.
\end{itemize}

Once the effective crossover probability $p$ has been estimated,
performance analysis becomes independent of the particular language model
or text corpus. Quantities such as the probability of successful
recovery, the required text length, redundancy, and resistance to
editing attacks can all be derived analytically. This provides guarantees
that remain valid for every watermarking system exhibiting the same
effective error probability.

The BSC model should therefore be viewed as an abstraction rather than a
description of a particular LLM. Different language models, embedding
algorithms and editing strategies may produce different values of $p$, but
their watermarking performance can nevertheless be compared within the
same mathematical framework. This significantly reduces the need for
repeated large-scale experiments whenever a new LLM or decoding algorithm
appears.

Naturally, the Binary Symmetric Channel is only the simplest member of a
hierarchy of increasingly realistic statistical models. Future work may
consider channels with memory, such as the Gilbert--Elliott
model~\cite{Gilbert1960,Elliott1963,Mushkin1989}, Hidden Markov
Models~\cite{Rabiner1989}, or explicit insertion/deletion
channels~\cite{Mitzenmacher2009,Morozov2024}. Such models may capture
burst errors or context-dependent editing more accurately. Nevertheless,
the BSC already provides a remarkably useful approximation for analysing
synchronisation-free watermarking schemes, because all editing operations
ultimately manifest themselves as independent errors in the recovered
parity bits.

We believe that developing statistical channel models for LLM
watermarking represents an important research direction. Such models would
provide a principled methodology for comparing watermarking schemes,
predicting their robustness, and analysing new generations of language
models without requiring exhaustive experimental evaluation for every
individual LLM.

\section{Extension to Diffusion Language Models}
\label{sec:diffusion}

Diffusion Language Models (DLMs) constitute an alternative paradigm for
text generation. Unlike autoregressive language models, which generate
tokens sequentially from left to right, diffusion models begin with a
corrupted (or masked) sequence and iteratively refine it until a coherent
sentence is obtained. During each refinement step, a subset of tokens is
updated while the remaining tokens remain unchanged. The process
terminates once the generated sequence becomes stable.
The synchronization-free watermark proposed in this paper naturally
extends to diffusion models. This is because every watermark bit depends
only on a pair of neighbouring tokens and not on the global history of
generation. Consequently, watermark constraints can be verified and
enforced locally throughout the diffusion process.

Let
$
T=(t_1,t_2,\ldots,t_N)
$
denote the current sequence of tokens produced by the diffusion model.
As in the autoregressive construction, the secret identity
$
S\in\mathcal S
$
is represented by the Reed--Solomon polynomial
$
f_S(x)=
s_0+s_1x+\cdots+s_{m-1}x^{m-1},
$
where $m$ denotes the identity length.
For every neighbouring pair $(t_i,t_{i+1})$ we compute
$
\alpha_i
=
H(K,\mathrm{id}(t_i))
\bmod q,
$
where $H$ is a cryptographic hash function and $K$ is the secret
watermarking key.
The corresponding watermark bit equals
$
w_i
=
f_S(\alpha_i)\bmod2.
$
Exactly as in the autoregressive case, the vocabulary is partitioned into
two colour classes,
$V_0$
and
$V_1,$
and token $t_{i+1}$ is said to satisfy the watermark constraint whenever
$
t_{i+1}\in V_{w_i}.
$
Since every watermark constraint depends only upon two consecutive tokens,
the watermarking problem becomes a collection of local consistency
conditions.


Every interior token participates simultaneously in two neighbouring
constraints.
\begin{description}
\item[\rm Left constraint.]
The colour of token $t_i$ must agree with the watermark bit determined by
its left neighbour,
$
t_i
\in
V_{\,f_S(H(K,\mathrm{id}(t_{i-1})))\bmod2}.
$
\item[\rm Right constraint.]
The colour of token $t_{i+1}$ must agree with the watermark bit generated
from the current token,
$
t_{i+1}
\in
V_{\,f_S(H(K,\mathrm{id}(t_i)))\bmod2}.
$
\end{description}
An interior token is therefore classified into one of four states.

\begin{center}
\begin{tabular}{|c|c|c|}
\hline
Left constraint & Right constraint & State\\
\hline
Satisfied & Satisfied & Fully consistent\\
Satisfied & Violated  & Left consistent\\
Violated  & Satisfied & Right consistent\\
Violated  & Violated  & Inconsistent\\
\hline
\end{tabular}
\end{center}
Only fully consistent tokens are guaranteed to remain unchanged throughout
subsequent refinement iterations.



The probabilistic analysis developed in the remainder of this section is
based on the following assumptions.

\fbox{%
\parbox{0.9\linewidth}{%
\begin{assumption}
\label{ass:probability}
The colour selected for each generated token is treated as an
independent random variable.
The probability that the \emph{left constraint} is satisfied is denoted
by $\varepsilon$, where
$
\frac12 \le \varepsilon \le 1.
$
The parameter $\varepsilon$ models the effectiveness of the watermark
embedding mechanism. Modern LLM watermarking algorithms bias the logits
towards the desired colour class before sampling the next token.
Consequently, whenever both colour classes contain sufficiently many
candidate tokens, the probability of selecting a token from the desired
colour class can be made arbitrarily close to one. On the other hand,
there are rare situations in which the language model offers no suitable
token of the required colour, making successful embedding impossible.
For this reason, $\varepsilon$ should be regarded as an average success
probability over all generated tokens.
Throughout the paper we analyse the algorithms for an arbitrary
$\varepsilon$, while the conservative choice
$$
\varepsilon=\frac12
$$
corresponds to unbiased sampling and therefore represents the worst
practical case.
The \emph{right constraint} depends only on the cryptographic hash
function. Under the standard assumption that the hash behaves as a
pseudorandom function, its output is uniformly distributed, giving
$$
\Pr[\text{right constraint holds}]
=\frac12.
$$
Furthermore, the left and right constraints are assumed to be
independent.
\end{assumption}
}}

\vspace{3mm}
It immediately follows that
\[
\Pr[\text{both constraints hold}]
=
\frac{\varepsilon}{2},
\mbox{  and  }
\Pr[\text{neither constraint holds}]
=
\frac{1-\varepsilon}{2}.
\]
These probabilities are independent of the identity length, the
particular language model, and the generated text. Their only dependence
is through the embedding efficiency parameter $\varepsilon$ and the
standard pseudorandomness assumption on the cryptographic hash function.


The diffusion process gradually transforms the sequence into one
satisfying all watermark constraints. To formalise this evolution, we
distinguish two classes of tokens.
\begin{description}
\item[\rm Committed token.]
A token is called \emph{committed} if both neighbouring constraints are
simultaneously satisfied. Once committed, the token is regarded as fixed
and is no longer modified by the watermarking algorithm.
\item[\rm Uncommitted token.]
A token is called \emph{uncommitted} whenever at least one neighbouring
constraint is violated. Such tokens remain eligible for further
refinement.
\end{description}
Let
$
U_k
$
denote the number of uncommitted tokens after the $k$-th diffusion
iteration.
The principal objective of watermark embedding is therefore to drive
$
U_k
\longrightarrow
0,
$
while preserving the semantic quality and fluency of the generated text.
The following subsections describe two algorithms achieving
this objective. The first algorithm commits tokens monotonically and is
particularly easy to analyse. The second algorithm additionally allows
local rearrangements of committed tokens, leading to substantially faster
convergence while preserving synchronization-free watermark recovery.

\subsection{Basic Commit Algorithm}
\label{sec:basic-commit}

The Basic Commit Algorithm is the simplest synchronization-free
watermarking strategy for diffusion language models. The algorithm scans
the sequence from left to right and commits a token as soon as its left
watermark constraint is satisfied. Once committed, a token is never
modified again and therefore serves as a permanent synchronization point
for all subsequent iterations.

Initially all tokens are uncommitted. During every diffusion iteration,
each uncommitted token is examined independently. If the watermark
constraint with its left neighbour is satisfied, the token becomes
permanently committed; otherwise it remains uncommitted and will be
reconsidered during the next iteration. Since only uncommitted tokens may
change, the number of committed tokens can only increase.
%
%
\begin{algorithm}[H]
\caption{Basic Commit Algorithm}
\label{alg:basic}
\begin{algorithmic}[1]
\Require Initial sequence of $N$ tokens
\Ensure Fully committed watermark
\State Mark every token as \emph{uncommitted}.
\Repeat
\For{$i=2,\ldots,N$}
\If{token $i$ is uncommitted}
\If{the left watermark constraint is satisfied}
\State Permanently commit token $i$.
\EndIf
\EndIf
\EndFor
\Until{no further commitments occur}
\end{algorithmic}
\end{algorithm}
%
\noindent
\underline{Probabilistic Analysis}
The analysis is based on Assumption~\ref{ass:probability}. Recall that
the probability that the left watermark constraint is satisfied equals
$\varepsilon$, where $\frac12\le\varepsilon\le1$. The parameter
$\varepsilon$ models the effectiveness of the watermark embedding
procedure and depends on the amount of logit bias applied by the language
model.
Let $U_t$ denote the number of uncommitted tokens after the $t$-th
diffusion iteration. During one iteration every currently uncommitted
token becomes committed independently with probability $\varepsilon$.
Consequently, it remains uncommitted with probability $1-\varepsilon$.
Therefore,
$
E[U_{t+1}]
=
(1-\varepsilon)E[U_t].
$
Since initially $U_0=N$,
$
E[U_t]
=
N(1-\varepsilon)^t.
$
Thus the expected number of uncommitted tokens decreases exponentially.
\begin{theorem}
Assume that the events determining whether the left watermark constraint
holds are mutually independent and occur with probability
$\varepsilon>0$. Then after $t$ diffusion iterations,
$
E[U_t]
=
N(1-\varepsilon)^t.
$
Consequently,
$
\lim_{t\rightarrow\infty}E[U_t]=0,
$
and every token is eventually committed with probability one.
\end{theorem}
\begin{proof}
Each uncommitted token survives one diffusion iteration with probability
$1-\varepsilon$. Hence
$
E[U_{t+1}]
=
(1-\varepsilon)E[U_t].
$
Repeated substitution gives
$
E[U_t]
=
N(1-\varepsilon)^t.
$
Since $0<1-\varepsilon<1$,
$
(1-\varepsilon)^t\rightarrow0
\qquad\text{as}\qquad
t\rightarrow\infty.
$
Therefore,
$
E[U_t]\rightarrow0.
$
Furthermore, the waiting time until an individual token becomes
committed follows a geometric distribution with parameter
$\varepsilon$. Its expected value equals
$
E[T]=\frac1{\varepsilon},
$
which is finite whenever $\varepsilon>0$. Consequently every token is
committed almost surely, implying convergence of the algorithm.
\hfill$\Box$
\end{proof}
The Basic Commit Algorithm has three important properties.
\begin{itemize}
\item The number of committed tokens is monotonically increasing.
\item Once committed, a token is never modified again.
\item The convergence speed depends only on the embedding probability
$\varepsilon$. In particular, the expected number of diffusion
iterations required to reduce the number of uncommitted tokens to a
constant is approximately
$
\frac{\ln N}{-\ln(1-\varepsilon)}
\approx
\frac{\ln N}{\varepsilon},
$
where the approximation follows from
$\ln(1-\varepsilon)\approx-\varepsilon$ for moderate values of
$\varepsilon$.
\end{itemize}

The algorithm propagates commitment information only from left to right.
Consequently, long runs of initially uncommitted tokens may require many
diffusion iterations before becoming fully committed. This limitation
motivates the Refined Commit Algorithm presented in the next section,
where commitment may propagate from either direction.

\subsection{Refined Commit Algorithm}
\label{sec:refined-commit}
The Basic Commit Algorithm commits a token whenever the \emph{left}
constraint is satisfied. It ignores the right constraint completely.
Although remarkably simple, this strategy commits many correct tokens but
fails to exploit additional information available from the neighbouring
token. In particular, whenever the left constraint holds while the right
one fails, the algorithm cannot determine whether the inconsistency
originates from the current token or from its right neighbour.

The Refined Commit Algorithm addresses this limitation by examining both
neighbouring constraints. Whenever the left constraint holds, the current
token is committed as in the Basic Commit Algorithm. If the right
constraint also holds, no further action is required. Otherwise, the
algorithm assumes that the inconsistency is more likely to originate from
the neighbouring token and revokes its commitment, allowing it to be
reconsidered during a later iteration. In this way, local watermark
errors are gradually propagated through the text until they disappear.
Consequently,
$
\Pr[L]=\varepsilon,
\Pr[R]=\frac12; 
$
and therefore
$
\Pr[L\land R]=\frac{\varepsilon}{2};
$
$
\Pr[L\land\overline{R}]=\frac{\varepsilon}{2};
$
$
\Pr[\overline{L}\land R]=\frac{1-\varepsilon}{2};
$
$
\Pr[\overline{L}\land\overline{R}]
=\frac{1-\varepsilon}{2},
$
where $\Pr[L]$ and $\Pr[R]$ stands for probabilities that the left and right constraints hold.
Hence the probability that exactly one constraint is satisfied equals
$
\Pr[\text{exactly one}]
=
\frac12,
$
independently of $\varepsilon$.


Let
$
L_i=
\bigl(
H(K,t_{i-1})=c(t_i)
\bigr)
$
denote the left constraint and
$
R_i=
\bigl(
H(K,t_i)=c(t_{i+1})
\bigr)
$
the right constraint, where $c(t_i)$ stands for a colour of the token $t_i$.
The refined algorithm is given below.
\begin{algorithm}[H]
\caption{Refined Commit Algorithm}
\label{alg:refined}
\begin{algorithmic}[1]
\For{every uncommitted token $t_i$}
    \If{$L_i$ and $R_i$ hold}
        \State Commit $t_i$
    \ElsIf{$L_i$ holds and $R_i$ fails}
        \State Commit $t_i$
        \State Revoke commitment of $t_{i+1}$
    \ElsIf{$L_i$ fails and $R_i$ holds}
        \State Commit $t_i$
        \State Revoke commitment of $t_{i-1}$
    \Else
        \State Leave $t_i$ uncommitted
    \EndIf
\EndFor
\end{algorithmic}
\end{algorithm}
Unlike the Basic Commit Algorithm, the refined procedure actively repairs
local inconsistencies by moving them towards neighbouring tokens. The
watermark therefore gradually ``self-corrects'' during successive
iterations.

Let
$
C_t
$
denote the fraction of committed tokens after iteration $t$.
A token becomes committed whenever at least one of the two constraints
holds. Hence
$
\Pr[\text{commit}]
=
\Pr[L\cup R]
=
\varepsilon+\frac12-\frac{\varepsilon}{2}
=
\frac{1+\varepsilon}{2}.
$
Thus, during the first iteration,
$
C_1=\frac{1+\varepsilon}{2}.
$
For the conservative choice $\varepsilon=\frac12$,
$
C_1=\frac34,
$
whereas for nearly perfect watermark embedding
($\varepsilon\approx1$),
$
C_1\approx1.
$


Thus, during the first iteration, $C_1=\frac{1+\varepsilon}{2}$. For the conservative choice
$\varepsilon=\frac12$, $C_1=\frac34$, whereas for nearly perfect watermark embedding
($\varepsilon\approx1$), $C_1\approx1$. Recall from Section~\ref{sec:basic-commit} that the
Basic Commit Algorithm's first-round commitment probability is $\varepsilon$. The Refined
Commit Algorithm therefore commits more tokens than Basic already in the first round,
though the size of this advantage depends on $\varepsilon$, as the following theorem shows.

\begin{theorem}
Under the independence assumption, $C_1=\frac{1+\varepsilon}{2}$, and
$\frac34\le C_1\le1$, for every $\frac12\le\varepsilon\le1$. Moreover,
\[
C_1-\varepsilon=\frac{1-\varepsilon}{2},
\]
so the first-round advantage of the Refined Commit Algorithm over the Basic Commit
Algorithm is $\frac{1-\varepsilon}{2}$. This gain is maximal, equal to $\frac14$, at the
conservative choice $\varepsilon=\frac12$, and vanishes as $\varepsilon\to1$: when the
embedding bias is already strong, Basic's own commit probability approaches~$1$ and there is
little room left for the right-constraint repair mechanism to add value.
\end{theorem}

\begin{proof}
A token is committed whenever at least one neighbouring constraint is satisfied. By
inclusion--exclusion,
\[
\Pr[L\cup R] = \Pr[L]+\Pr[R]-\Pr[L\land R].
\]
Substituting $\Pr[L]=\varepsilon$, $\Pr[R]=\frac12$, $\Pr[L\land R]=\frac{\varepsilon}{2}$,
gives
\[
C_1 = \varepsilon+\frac12-\frac{\varepsilon}{2} = \frac{1+\varepsilon}{2}.
\]
Since $\frac12\le\varepsilon\le1$, we immediately obtain $\frac34\le C_1\le1$. Recalling that
the Basic Commit Algorithm's first-round commit probability is $\varepsilon$ (Theorem~2), the
gain of Refined over Basic is
\[
C_1-\varepsilon = \frac{1+\varepsilon}{2}-\varepsilon = \frac{1-\varepsilon}{2},
\]
which proves the claim.
\hfill$\square$
\end{proof}

The refined algorithm offers two advantages over the Basic Commit Algorithm. First, the
expected fraction of committed tokens after one iteration increases from $\varepsilon$ to
$\frac{1+\varepsilon}{2}$, a gain of $\frac{1-\varepsilon}{2}$. This gain is largest,
$\frac14$, at the conservative choice $\varepsilon=\frac12$, and shrinks toward $0$ as
$\varepsilon\to1$: when the embedding bias is already strong, Basic's own commit rate is
close to its ceiling and there is little room left for the right-constraint repair
mechanism to add value. Second, local watermark inconsistencies are not merely ignored but
are actively propagated towards neighbouring tokens. Consequently, repeated iterations
gradually eliminate isolated errors and improve the density of committed tokens, with the
largest benefit realised precisely when the embedding bias $\varepsilon$ is weak. The
convergence behaviour of this iterative correction process is analysed in the following
subsection.

\subsection{Sliding Commit Algorithm}
\label{sec:sliding}

The Basic Commit Algorithm commits a token whenever its left watermark
constraint is satisfied. The Refined Commit Algorithm additionally
corrects local inconsistencies by revoking neighbouring committed tokens.
Both algorithms are monotone in the sense that, apart from local
corrections, the number of committed tokens generally increases during
successive iterations.
The following algorithm adopts a different philosophy. Instead of fixing
incorrect commitments locally, it allows the commitment boundary itself
to move through the text. Consequently, committed tokens may become
uncommitted while neighbouring uncommitted tokens become committed.
The commitment therefore ``slides'' until a stable configuration is
reached.


Suppose that an uncommitted token $t_i$ lies between two committed
tokens $t_{i-1}$ and $t_{i+1}$.
The algorithm proceeds as follows.
\begin{enumerate}
\item
If both watermark constraints hold, commit $t_i$.
\item
If only the left constraint holds, commit $t_i$ and revoke the
commitment of $t_{i+1}$.
\item
If only the right constraint holds, commit $t_i$ and revoke the
commitment of $t_{i-1}$.
\item
If neither constraint holds, no modification is made.
\end{enumerate}
Thus every successful commitment either enlarges the committed region
or shifts it by one position. 
\begin{algorithm}[H]
\caption{Sliding Commit Algorithm}
\label{alg:sliding}
\begin{algorithmic}[1]
\Repeat
\For{every uncommitted token $t_i$ whose neighbours are committed}
    \If{left and right constraints hold}
        \State Commit $t_i$
    \ElsIf{left constraint holds}
        \State Commit $t_i$
        \State Uncommit $t_{i+1}$
    \ElsIf{right constraint holds}
        \State Commit $t_i$
        \State Uncommit $t_{i-1}$
    \EndIf
\EndFor
\Until{no further changes occur}
\end{algorithmic}
\end{algorithm}


Assume the probabilistic model introduced in
Assumption~\ref{ass:probability}. Then
$
\Pr[L]=\varepsilon,
\Pr[R]=\frac12,
$
where the left and right constraints are mutually independent.
Consequently,
$
\Pr[L\wedge R]
=
\frac{\varepsilon}{2},
$
$
\Pr[L\wedge\overline{R}]
=
\frac{\varepsilon}{2},
$
$
\Pr[\overline{L}\wedge R]
=
\frac{1-\varepsilon}{2},
$
and
$
\Pr[\overline{L}\wedge\overline{R}]
=
\frac{1-\varepsilon}{2}.
$
Therefore,
$
\Pr[\text{a sliding operation occurs}]
=
\Pr[L\wedge\overline{R}]
+
\Pr[\overline{L}\wedge R]
=
\frac12.
$
Remarkably, the probability of moving the commitment boundary is
independent of $\varepsilon$.
%
\begin{theorem}
Under Assumption~\ref{ass:probability}, every active boundary advances by one position with
probability
$
\frac12,
$
during each iteration. Consequently, the expected displacement after
$t$ iterations equals
$
\frac{t}{2},
$
while the variance equals
$
\frac{t}{4}.
$
Thus the commitment boundary performs a biased random walk whose drift
is independent of the watermark embedding probability
$\varepsilon$.
\end{theorem}

\begin{proof}
A boundary movement occurs precisely when exactly one of the two
constraints holds.
Since
$
\Pr[L\wedge\overline{R}]
=
\frac{\varepsilon}{2}
$
and
$
\Pr[\overline{L}\wedge R]
=
\frac{1-\varepsilon}{2},
$
their sum equals
$
\frac{\varepsilon}{2}
+
\frac{1-\varepsilon}{2}
=
\frac12.
$
Each successful movement shifts the boundary by one token.
The number of successful movements after $t$ iterations therefore
follows the binomial distribution
$
X_t\sim\mathrm{Binomial}
\left(
t,\frac12
\right).
$
Hence
$
E[X_t]
=
\frac{t}{2},
$
and
$
\mathrm{Var}(X_t)
=
\frac{t}{4},
$
which proves the theorem.
\hfill$\Box$
\end{proof}


Unlike the Basic and Refined Commit Algorithms, the Sliding Commit
Algorithm does not attempt to preserve every existing commitment.
Instead, commitments migrate through the text until neighbouring
constraints become mutually consistent. In this respect the algorithm
resembles local optimisation methods, where the current solution is
allowed to deteriorate temporarily in order to obtain a better global
configuration.
Another attractive property is that the expected speed of propagation
is independent of the embedding quality $\varepsilon$. Increasing the
logit bias improves the probability that the correct token colour is
generated, but does not change the expected rate at which commitment
boundaries move through the document.

\subsection{Comparison of Commit Algorithms}
\label{sec:commit-comparison}

Table~\ref{tab:commit-comparison} summarises the three commit
algorithms introduced above. All three share the same underlying
constraint model ($\Pr[L]=\varepsilon$, $\Pr[R]=\tfrac12$, independent)
but differ in which constraints trigger a commit, whether existing
commitments can be revoked, and where in the sequence they operate.
\begin{table}[ht]
\centering
\caption{Comparison of the three commit algorithms for diffusion watermarking.}
\label{tab:commit-comparison}
{\tiny
\begin{tabular}{lccc}
\hline
\textbf{Property} & \textbf{Basic} & \textbf{Refined} & \textbf{Sliding} \\
\hline
Commit rule &
$L$ holds &
$L\lor R$ holds &
boundary token, $L$ or $R$ holds \\
Right constraint used &
no &
yes (repair only) &
yes (repair + commit) \\
Can revoke commitments &
no &
yes, on the neighbour &
yes, on one neighbour \\
Applies to &
every uncommitted token &
every uncommitted token &
uncommitted tokens flanked by committed neighbours \\
Monotonic in $|{\rm committed}|$ &
yes &
no (local repair) &
no (boundary shifts) \\
First-iteration commit prob.\ &
$\varepsilon$ &
$(1+\varepsilon)/2$ &
$\tfrac12$ (boundary advances) \\
Depends on $\varepsilon$ &
yes, directly &
only through the gain term &
no \\
Convergence measure &
$E[U_t]=N(1-\varepsilon)^t$ &
not closed-form; errors self-correct &
$E[\text{displacement}]=t/2$, $\mathrm{Var}=t/4$ \\
Iterations to convergence &
$O(\ln N/\varepsilon)$ &
slower than Basic (empirical) &
$O(N)$ \\
Propagation direction &
left-to-right only &
bidirectional, local &
bidirectional, boundary walk \\
\hline
\end{tabular}
}
\end{table}

\paragraph{Discussion.}
The three algorithms occupy distinct points on a speed/robustness
trade-off. The Basic Commit Algorithm is the simplest and fastest to
analyse: since it never revokes a commitment, its expected number of
uncommitted tokens decays exponentially, $E[U_t]=N(1-\varepsilon)^t$,
requiring only $O(\ln N/\varepsilon)$ iterations to stabilise. Its
weakness is that it ignores the right constraint entirely, so an
incorrectly committed token can never be corrected once fixed.

The Refined Commit Algorithm keeps the same commit condition on $L$
but additionally uses $R$ to decide whether the \emph{neighbouring}
token's commitment should be revoked. 
This raises the first-iteration commit probability from $\varepsilon$ to $(1+\varepsilon)/2$
--- a gain of $\frac{1-\varepsilon}{2}$, largest at the conservative choice $\varepsilon=\frac12$
and shrinking to $0$ as $\varepsilon\to1$, since $\Pr[\text{exactly one constraint holds}]=\frac12$
is itself independent of $\varepsilon$ even though the resulting gain over Basic is not.

The Sliding Commit Algorithm is qualitatively different: rather than
committing tokens throughout the sequence, it operates only at the
boundary between committed and uncommitted regions, and a successful
step moves that boundary by exactly one token. Because the probability
of such a step is $\tfrac12$ independently of $\varepsilon$, the
boundary behaves as a driftless-in-$\varepsilon$ random walk with
$E[\text{displacement}]=t/2$ and $\mathrm{Var}=t/4$. The
practical consequence is that Sliding needs $O(N)$ iterations for a
single boundary to traverse the sequence --- asymptotically slower
than Basic's $O(\ln N/\varepsilon)$ --- but it guarantees that every
local inconsistency is eventually pushed out of the committed region
rather than being frozen in place, which neither Basic nor Refined
guarantees on its own.

In short: Basic is fastest but cannot repair errors; Refined adds a
fixed, $\varepsilon$-independent improvement in commit rate together
with local self-correction; Sliding sacrifices asymptotic speed for a
guarantee of eventual consistency along the commitment boundary. A
practical implementation may combine Refined commits in the interior
of the sequence with Sliding-style boundary repair to obtain both the
faster convergence rate of Refined and the correctness guarantee of
Sliding.

\begin{example}
We compare the number of denoising iterations required to fully
commit a sequence of $n$ tokens under the Basic, Refined, and
Sliding commit algorithms, fixing the embedding-bias strength at
$\varepsilon=0.75$. For each algorithm we report $t_{0.90}$ and
$t_{0.99}$, the smallest number of iterations $t$ such that the
whole sequence is committed with probability at least $0.90$ and
$0.99$ respectively.
Table~\ref{tab:sim-comparison} reports $t_{0.90}$ and $t_{0.99}$
for $n\in\{32,64,128,512,1024\}$.
\begin{table}[ht]
\centering
\caption{Simulated/exact number of denoising iterations $t_p$
required to commit the entire sequence with probability $p$, for
$\varepsilon=0.75$. Basic and Sliding are computed exactly;
Refined is estimated from $4{,}000$ Monte Carlo trials.}
\label{tab:sim-comparison}
\begin{tabular}{r rrr rrr}
\hline
 & \multicolumn{3}{c}{$p=0.90$} & \multicolumn{3}{c}{$p=0.99$} \\
\cline{2-4}\cline{5-7}
$n$ & Basic & Refined & Sliding & Basic & Refined & Sliding \\
\hline
$32$   & $5$ & $11$ & $72$    & $6$ & $16$ & $82$    \\
$64$   & $5$ & $13$ & $141$   & $7$ & $18$ & $154$   \\
$128$  & $6$ & $15$ & $275$   & $7$ & $20$ & $293$   \\
$512$  & $7$ & $18$ & $1063$  & $8$ & $23$ & $1099$  \\
$1024$ & $7$ & $19$ & $2104$  & $9$ & $25$ & $2153$  \\
\hline
\end{tabular}
\end{table}
Three patterns emerge.
Basic scales logarithmically in $n$, exactly as predicted by
$t_p=O(\ln N/\varepsilon)$: moving from $n=32$ to $n=1024$ (a
$32\times$ increase) raises $t_{0.90}$ from only $5$ to $7$
iterations.

Sliding scales linearly in $n$, and is dramatically slower
than the other two algorithms for large sequences: at $n=1024$ it
requires roughly $2100$ iterations versus fewer than $20$ for
Basic or Refined. This is a direct, quantitative consequence of
Theorem~\ref{thm:sliding} --- an $\varepsilon$-independent advance
probability of $\tfrac12$ forces $\Theta(n)$ iterations to sweep
a single boundary across the sequence, regardless of how strong
the watermark bias is. Sliding is therefore unsuitable as a
stand-alone algorithm for long sequences; its value lies in local
self-repair rather than driving full-sequence convergence.

Refined is slower than Basic at every value of $n$ tested,
which at first appears to contradict
Section~\ref{sec:refined-commit}'s observation that Refined's
per-token, per-iteration commit probability, $(1+\varepsilon)/2$,
exceeds Basic's $\varepsilon$. The resolution is that Refined's
advantage is local and per-token, while full-sequence completion
requires \emph{all} $n-1$ tokens to be \emph{simultaneously and
durably} committed. Because a token can be revoked by a neighbour
whose own right-constraint check fails purely by chance
(probability $\tfrac12$, independent of correctness), already-valid
commitments face a persistent, $\varepsilon$-independent hazard of
being undone. For a single token this repair mechanism is a net
benefit, but across a long sequence it introduces churn that
delays the moment at which the \emph{entire} sequence is
simultaneously stable, offsetting Refined's higher instantaneous
commit rate. This suggests that Refined is best used together with
Basic-style permanent commitment once a token's constraints have
been verified for several consecutive iterations, rather than
being applied uniformly throughout the diffusion process.
\end{example}

\subsection{Future Research}
\label{sec:diffusion-future}

\paragraph{Watermark-aware denoising.}
One may seek denoising strategies that maximise the number of
satisfied congruences at every step while preserving text quality,
potentially reducing denoising steps.

\paragraph{Variable-length blocks.}
If the denoising mechanism permits insertion of dummy tokens or
deletion of tokens, the synchronisation-free property is preserved
by assigning dummy tokens a reserved identifier contributing no
watermark observation. Their insertion or removal affects neither
evaluation points nor embedded bits of surrounding pairs.

\paragraph{Theoretical convergence analysis.}
A rigorous bound on the number of denoising steps as a function of
$n$, watermark density, and the LLM's token distribution remains
an important open problem. The analysis above suggests that
diffusion LLMs require only $O(\log n)$ denoising steps to embed
an $n$-bit watermark, which is negligible relative to typical
diffusion generation budgets.

\section{Conclusions and Future Research}

We have presented a new synchronization-free watermarking scheme for Large Language Models based on binary congruence constraints derived from Reed--Solomon polynomial evaluations. Unlike existing watermarking approaches, which typically embed consecutive blocks of watermark bits, the proposed method associates every watermark with a single token pair. A cryptographic hash computed from one token determines the evaluation point, while the corresponding parity bit is embedded into the following token. Consequently, every watermark is locally verifiable and independent of the surrounding text, making the scheme naturally robust against insertions, deletions, substitutions, and token reordering.

The paper provides the mathematical foundations of the proposed construction, including probabilistic analysis of binary congruence systems, sufficient conditions for reliable recovery, several recovery algorithms covering different identity sizes, and an error analysis based on the Binary Symmetric Channel model. The obtained results indicate that only a modest redundancy is required to recover the embedded identity with high probability even when a significant fraction of watermark bits is corrupted.

An important advantage of the proposed framework is its scalability. Although the paper considers embedding a single secret identity represented by one Reed--Solomon codeword, the construction naturally generalises to arbitrarily long identities. A long secret can be partitioned into multiple short segments (for example, 32- or 64-bit blocks), each protected by an appropriately sized Reed--Solomon code. During embedding, the watermark corresponding to each token pair is determined by a small number of bits extracted from a cryptographic hash, which pseudorandomly selects both the identity segment and the evaluation point within its associated code. As a result, information belonging to different parts of the secret is uniformly scattered throughout the generated text rather than stored in contiguous blocks. This randomisation improves robustness against local text modifications and preserves the synchronization-free nature of the scheme.

Several research directions remain open. The most important theoretical problem is the design of optimal collections of binary congruence constraints that minimise the number of transmitted watermark bits while guaranteeing unique recovery. It would also be interesting to investigate decoding algorithms based on SAT/SMT solvers, lattice reduction, and probabilistic inference for very large identity spaces. Finally, practical implementation within modern LLM decoding pipelines, together with extensive evaluation against paraphrasing attacks, machine translation, summarisation, and adversarial editing, will provide further insight into the practical security and robustness of the proposed watermarking framework.

{\footnotesize
\bibliographystyle{plain}
\bibliography{LLM}

@article{qu2024provably,
  author  = {Qu, Wenjie and Zheng, Wengrui and Tao, Tianyang and Yin, Dong and
             Jiang, Yanze and Tian, Zhihua and Zou, Wei and Jia, Jinyuan and
             Zhang, Jiaheng},
  title   = {Provably Robust Multi-bit Watermarking for {AI}-generated Text},
  journal = {arXiv preprint arXiv:2401.16820},
  year    = {2025},
  note    = {v5, 28 Jan 2025},
  url     = {https://arxiv.org/abs/2401.16820}
}

@inproceedings{Zhao_2023,
  author    = {Xuandong Zhao and Prabhanjan Vijendra Ananth and Lei Li and Yu-Xiang Wang},
  title     = {Provable Robust Watermarking for AI-Generated Text},
  booktitle = {The Twelfth International Conference on Learning Representations (ICLR)},
  year      = {2024},
  url       = {https://openreview.net/forum?id=SsmT8aO45L}
}

@inproceedings{Niess_2025,
  author    = {Georg Niess and Roman Kern},
  title     = {Ensemble Watermarks for Large Language Models},
  booktitle = {Proceedings of the 63rd Annual Meeting of the Association for Computational Linguistics (ACL)},
  year      = {2025},
  doi       = {10.18653/v1/2025.acl-long.145},
  url       = {https://aclanthology.org/2025.acl-long.145/}
}

@article{Nature2024SynthID,
  author  = {Pushmeet Kohli and others},
  title   = {SynthID-Text: Practical Watermarking for Large Language Models},
  journal = {Nature},
  year    = {2024}
}

@article{Kuditipudi2024,
  author  = {Rohith Kuditipudi and others},
  title   = {Robust and Semantically Invariant Watermarks for Large Language Models},
  journal = {ICLR},
  year    = {2024}
}

@book{cox2007,
  author    = {Ingemar J. Cox and Matthew L. Miller and Jeffrey A. Bloom and Jessica Fridrich and Ton Kalker},
  title     = {Digital Watermarking and Steganography},
  publisher = {Morgan Kaufmann},
  edition   = {2},
  year      = {2007}
}

@article{Morozov2024,
  author  = {Ruslan Morozov and Tolga M. Duman},
  title   = {Markov Insertion/Deletion Channels: Information Stability and Capacity Bounds},
  journal = {arXiv preprint arXiv:2401.16063},
  year    = {2024}
}

@article{Mitzenmacher2009,
  author  = {Michael Mitzenmacher},
  title   = {A Survey of Results for Deletion Channels and Related Synchronization Channels},
  journal = {Probability Surveys},
  volume  = {6},
  pages   = {1--33},
  year    = {2009},
  doi     = {10.1214/08-PS141}
}

@book{Rabiner1989,
  author    = {Lawrence R. Rabiner},
  title     = {A Tutorial on Hidden Markov Models and Selected Applications in Speech Recognition},
  journal   = {Proceedings of the IEEE},
  volume    = {77},
  number    = {2},
  pages     = {257--286},
  year      = {1989},
  doi       = {10.1109/5.18626}
}

@article{Mushkin1989,
  author  = {Mordechai Mushkin and Isaac Bar-David},
  title   = {Capacity and Coding for the Gilbert--Elliott Channels},
  journal = {IEEE Transactions on Information Theory},
  volume  = {35},
  number  = {6},
  pages   = {1277--1290},
  year    = {1989},
  doi     = {10.1109/18.45284}
}

@article{Elliott1963,
  author  = {E. O. Elliott},
  title   = {Estimates of Error Rates for Codes on Burst-Noise Channels},
  journal = {Bell System Technical Journal},
  volume  = {42},
  number  = {5},
  pages   = {1977--1997},
  year    = {1963},
  doi     = {10.1002/j.1538-7305.1963.tb00955.x}
}

@article{Gilbert1960,
  author  = {Edgar N. Gilbert},
  title   = {Capacity of a Burst-Noise Channel},
  journal = {Bell System Technical Journal},
  volume  = {39},
  number  = {5},
  pages   = {1253--1265},
  year    = {1960},
  doi     = {10.1002/j.1538-7305.1960.tb03959.x}
}

@misc{Fu_2024,
	author = {Jiayi Fu and Xuandong Zhao and Ruihan Yang and Yuansen Zhang and Jiangjie Chen and Yanghua Xiao},
	howpublished = {ArXiv eprint 2402.12948},
	month = {02},
	title = {GumbelSoft: Diversified Language Model Watermarking via the GumbelMax-trick},
	url = {https://arxiv.org/pdf/2402.12948.pdf},
	year = {2024}}

@misc{wong_2025,
	author = {Ka Him Wong and Jicheng Zhou and Jiantao Zhou and Yain-Whar Si},
	title = {An End-to-End Model For Logits Based Large Language Models Watermarking},
	url = {https://openreview.net/forum?id=0KHW6yXdiZ},
	year = {2025}}

@inproceedings{Kirchenbauer2023Watermark,
  author    = {John Kirchenbauer and
               Jonas Geiping and
               Yuxin Wen and
               Jonathan Katz and
               Ian Miers and
               Tom Goldstein},
  title     = {A Watermark for Large Language Models},
  booktitle = {Proceedings of the 40th International Conference on Machine Learning (ICML 2023)},
  series    = {Proceedings of Machine Learning Research},
  volume    = {202},
  pages     = {17061--17084},
  year      = {2023},
  editor    = {Andreas Krause and
               Emma Brunskill and
               Kyunghyun Cho and
               Barbara Engelhardt and
               Sivan Sabato and
               Jonathan Scarlett},
  publisher = {PMLR},
  url       = {https://proceedings.mlr.press/v202/kirchenbauer23a.html}
}

@misc{Zhu_2024,
	author = {Chaoyi Zhu and Jeroen Galjaard and Pin-Yu Chen and Lydia Y. Chen},
	howpublished = {ArXiv eprint 2403.13000},
	month = {03},
	title = {Duwak: Dual Watermarks in Large Language Models},
	url = {https://arxiv.org/pdf/2403.13000.pdf},
	year = {2024}}

@misc{Xu_2024,
	author = {Zhenyu Xu and Kun Zhang and Victor S. Sheng},
	eprint = {2410.10876},
	howpublished = {ArXiv eprint 2410.10876},
	month = {10},
	title = {FreqMark: Frequency-Based Watermark for Sentence-Level Detection of LLM-Generated Text},
	url = {https://arxiv.org/pdf/2410.10876.pdf},
	year = {2024}}

@misc{Cui_2025,
	author = {Xinyue Cui and Johnny Tian-Zheng Wei and Swabha Swayamdipta and Robin Jia},
	howpublished = {ArXiv eprint 2503.04036},
	month = {03},
	title = {Robust Data Watermarking in Language Models by Injecting Fictitious Knowledge},
	url = {https://arxiv.org/pdf/2503.04036.pdf},
	year = {2025}}

@misc{Liu_2023,
	author = {Yixin Liu and Hongsheng Hu and Xun Chen and Xuyun Zhang and Lichao Sun},
	eprint = {2305.13257},
	howpublished = {arXiv preprint arXiv:2302.13971},
	month = {05},
	title = {Watermarking Text Data on Large Language Models for Dataset Copyright},
	url = {https://arxiv.org/pdf/2305.13257.pdf},
	year = {2023}}

@misc{Sander_2024,
	author = {Tom Sander and Pierre Fernandez and Alain Durmus and Matthijs Douze and Teddy Furon},
	eprint = {2402.14904},
	howpublished = {ArXiv eprint},
	month = {02},
	title = {Watermarking Makes Language Models Radioactive},
	url = {https://arxiv.org/pdf/2402.14904.pdf},
	year = {2024}}

@article{touvron_2023,
	author = {Hugo Touvron and Thibault Louvrier and Matthieu Cord and Piotr Bojanowski and Edouard Grave and Guillaume Lample},
	journal = {arXiv preprint arXiv:2302.13971},
	title = {{LLaMA}: Open and Efficient Foundation Language Models},
	url = {https://arxiv.org/abs/2302.13971},
	year = {2023}}

@inproceedings{chowdhery_2022,
	author = {Aakanksha Chowdhery and Anish Vaswani and Steven J. Rennie and Mikhail Pavlov and Jacob Devlin and Sanjay Aggarwal and Mike Lewis and Neil Houlsby and Colin Raffel and Barbara Plank and Lee Howard and Martin D. Riley and Michael Swietojanski and Mo Yu and Dipanjan Das and Mike Schuster and Yiming Yang and Jakob Uszkoreit and Yonghui Wu},
	booktitle = {Proceedings of the 39th International Conference on Machine Learning (ICML 2022)},
	title = {{PaLM}: Scaling Language Modeling with Pathways},
	url = {https://arxiv.org/abs/2204.02311},
	year = {2022}}

@misc{openai2023gpt4,
	author = {{OpenAI}},
	howpublished = {{\url{https://openai.com/research/gpt-4}}},
	title = {{GPT-4 Technical Report}},
	year = {2023}}
}

\end{document}